\documentclass[journal]{IEEEtran}
\usepackage{algpseudocode}

\newcounter{mytempeqncnt}
\newtheorem{thm}{Theorem}
\newtheorem{lem}{Lemma}

\newenvironment{proof}[1][Proof]{\begin{trivlist}
\item[\hskip \labelsep {\bfseries #1}]}{\end{trivlist}}

\newcommand{\qed}{\nobreak \ifvmode \relax \else
      \ifdim\lastskip<1.5em \hskip-\lastskip
      \hskip1.5em plus0em minus0.5em \fi \nobreak
      \vrule height0.75em width0.5em depth0.25em\fi}

\ifCLASSINFOpdf
  \usepackage[pdftex]{graphicx}
  \graphicspath{{./}{Figs/}}
  \DeclareGraphicsExtensions{.pdf,.jpeg,.png}
\else
\fi
\usepackage{amssymb}
\usepackage{amsmath}
\usepackage{mathtools}
\usepackage{bigdelim}

\newcommand*{\TitleFont}{%
      \fontsize{16}{20}%
      \selectfont}

\begin{document}
\title{\TitleFont Achievable Degrees-of-Freedom of the $K$-user SISO Interference Channel with Blind Interference Alignment using Staggered Antenna Switching}

\author{Ahmed~M.~Alaa\IEEEauthorrefmark{1}\IEEEauthorrefmark{2},~\IEEEmembership{Student Member,~IEEE}, Mahmoud~H.~Ismail \IEEEauthorrefmark{2} \IEEEauthorrefmark{3},~\IEEEmembership{Senior Member,~IEEE} 
\thanks{
\IEEEauthorrefmark{1} Department of Electrical Engineering, University of California Los Angeles (UCLA), Los Angeles, CA, 90095, USA (e-mail: ahmedmalaa@ucla.edu).}
\thanks{
\IEEEauthorrefmark{2} Department of Electronics and Electrical Communications Engineering, Cairo University, Giza, 12613, Egypt (e-mails: \{aalaa, mismail\}@eece.cu.edu.eg).}
\thanks{
\IEEEauthorrefmark{3}Department of Electrical Engineering, The American University of Sharjah, PO Box 26666, Sharjah, UAE (email: mhibrahim@aus.edu).
}
}
\markboth{XXXXXX,~Vol.~XX, No.~X, XXXX~201X}%
{XX \MakeLowercase{\textit{et al.}}: Blind Interference Alignment for the $K$-user Interference Channel}

\maketitle
\begin{abstract}
In this paper, we characterize the achievable (linear) Degrees-of-Freedom (DoF) of the $K$-user SISO Interference Channel with Blind Interference Alignment (BIA) using staggered antenna switching. In such scheme, each transmitter is equipped with one conventional antenna and each receiver is equipped with one reconfigurable (multi-mode) antenna. Assuming that the channel is known to the receivers only, we show that by using linear BIA with staggered antenna switching at the receiver, a sum DoF of $\frac{2K}{K+2}$ is achievable. This result implies that using BIA, we can double the DoF achieved by orthogonal multiple access schemes for large number of users. Moreover, we propose a systematic algorithm to generate the transmit beamforming vectors and the reconfigurable antenna switching patterns, showing that it can achieve the $\frac{2K}{K+2}$ sum DoF for any $K$. Finally, we apply this algorithm to the 4-user SISO Interference Channel, and demonstrate that $\frac{4}{3}$ sum DoF is indeed achievable.
\end{abstract}


\IEEEpeerreviewmaketitle
\section{Introduction}
\IEEEPARstart{I}{nterference} is one of the most important factors that limit the capacity of wireless networks. Because characterizing the exact capacity region of interference channels is quite complicated, the Degrees-of-Freedom (DoF) has emerged as a convenient metric to quantify the capacity limits of such channels. The DoF correspond to the number of independent interference-free signaling dimensions, and have been characterized for many multiuser network settings.

Recently, {\it interference alignment} (IA) has been shown to achieve the DoF of various wireless networks \cite{1x},\cite{1}. In \cite{2}, Cadambe {\it et al.} have shown that IA can achieve the $\frac{K}{2}$ DoF of the $K$-user interference channel. However, the IA schemes presented therein require global channel state information to be known at the transmitters (CSIT), which is not always possible in practical systems. Stemming from this point, many research efforts have been dedicated to develop ``{\it Blind Interference Alignment"} (BIA) techniques that can be applied without CSIT. In \cite{3}, Jafar has shown that BIA can be applied if the channel temporal correlation structure follows specific patterns. Such patterns were created artificially using reconfigurable antennas in \cite{4} for the broadcast channel (BC). However, the channel temporal correlation structure that achieves the DoF of the $K$-user interference channel cannot arise in nature \cite{3}, and cannot be created using reconfigurable antennas \cite{4}.

Recently in \cite{7}, Wang presented the first characterization for the sum DoF when using BIA over the 3-user interference channel. It was shown that BIA can achieve a sum DoF of $\frac{6}{5}$, which is greater than the unity sum DoF achieved by naive orthogonal schemes. Besides, it was also shown that $\frac{6}{5}$ is the upper bound on the DoF achieved by linear BIA schemes. However, this result is specific for the 3-user channel, and the DoF of the generic $K$-user interference channel with BIA remains an open issue. In this work, we extend the analysis in \cite{7}, and show that when using BIA with staggered antenna switching over a $K$-user interference channel, a sum DoF of $\frac{2K}{K+2}$ is achievable. A key insight behind this result is that any signal vector can be aligned with interference at $K-2$ unintended receivers, which corresponds to the alignment scheme considered in this paper. Such result suggests that BIA can achieve double the rate achieved by orthogonal schemes for large $K$. This result is novel as it is the first characterization for the DoF of the $K$-user interference channel with BIA. Moreover, we propose a systematic algorithm for generating the transmit beamforming vectors and the receive antennas switching patterns in order to achieve the $\frac{2K}{K+2}$ DoF. Finally, we apply this algorithm to the 4-user SISO interference channel showing that $\frac{4}{3}$ sum DoF are achievable almost surely.

It is worth mentioning that several works have characterized the DoF of the interference channel with delayed CSIT. For example, in \cite{8}, it was shown that the DoF of the $K$-user interference channel with IA and outdated CSIT approach the limiting value $4/(6 \, \ln(2)  - 1) \approx 1.266$ as $K\rightarrow \infty$. Also, the sum DoF of two ergodic IA techniques with delayed CSIT presented in \cite{9} was shown to be $\frac{2K}{K+2}$. It is important to note here that our work departs from these works as BIA does not require any CSIT. Other BIA schemes, such as the relay-assisted scheme in \cite{10}, requires the presence of several relays between transmitters and receivers, which may not be the case in many practical scenarios. Contrarily, BIA using staggered antenna switching only requires the usage of reconfigurable antennas at the receivers, which does not entail significant hardware complexity as reconfigurable antennas use a single RF chain \cite{4}.

The rest of the paper is organized as follows. In Section II, we present the system model. Next, in Section III, we derive the sum DoF of the $K$-user SISO interference channel with BIA using staggered antenna switching at the receiver. A systematic algorithm to generate the beamforming vectors and the antenna switching patterns is presented in Section IV. This algorithm is applied for the 4-user SISO Interference Channel, and we show that $\frac{4}{3}$ sum DoF is achievable. Finally, we draw our conclusion in Section V.

\section{System Model}

We consider a fully connected $K$-user SISO interference channel where each transmitter is equipped with one conventional antenna, and each receiver is equipped with a reconfigurable (multi-mode) antenna that can switch among $M$ distinct modes. The received signal at the $k^{th}$ receiver over $m$ channel uses is given by
\begin{equation}
\label{1}
{\bf y}_{k} = \sum_{i=1}^{K} {\bf H}_{ki} {\bf x}_{i} + {\bf n}_{k}, \,\, k, i \in \{1,2,...,K\}
\end{equation}
where ${\bf y}_{k} \in \mathbb{C}^{m \times 1}$ represents the received signal over $m$ channel uses (time or frequency slots), ${\bf x}_{i} \in \mathbb{C}^{m \times 1}$ is the transmitted signal vector by the $i^{th}$ user, ${\bf n}_{k} \in \mathbb{C}^{m \times 1}$ is an additive noise vector where the noise sample at the $j^{th}$ channel use is $n_{k}(j) \sim \mathcal{CN}(0,1)$, and ${\bf H}_{ki} \in \mathbb{C}^{m \times m}$ is the channel matrix between the $i^{th}$ transmitter and the $k^{th}$ receiver. The channel matrix can be written as
\begin{equation}
\label{2}
{\bf H}_{ki} = \mbox{diag}([h_{ki}({\bf p}_{k}(1)) \,\, h_{ki}({\bf p}_{k}(2)) \,\, ... \,\, h_{ki}({\bf p}_{k}(m))]),
\end{equation}
where ${\bf p}_{k}$ is the reconfigurable antenna switching pattern at the $k^{th}$ receiver. We assume an antenna with $M$ possible modes, thus ${\bf p}_{k}(j) \in \{1,2,...,M\}$ is the selected antenna mode at the $j^{th}$ channel use. Each mode corresponds to a distinct channel realization, which means that at the $j^{th}$ channel use, the channel between the $i^{th}$ transmitter and the $k^{th}$ receiver belongs to the set $\{h_{ki}(1), h_{ki}(2),..., h_{ki}(M)\}$, depending on the switching pattern ${\bf p}_{k}$. We assume that all channel coefficients are drawn i.i.d. and are constant for $m$ channel uses. The transmitted signal vector ${\bf x}_{i}$ is given by
\begin{equation}
\label{3}
{\bf x}_{i} = \sum_{d=1}^{d_{i}} s_{d}^{[i]} \, {\bf u}_{d}^{[i]}
\end{equation}
where $d_{i}$ is the number of symbols transmitted by the $i^{th}$ user over $m$ channel uses, $s_{d}^{[i]}$ is the $d^{th}$ transmitted symbol, and ${\bf u}_{d}^{[i]}$ is an $m \times 1$ transmit beamforming vector for the $d^{th}$ symbol, where ${\bf u}_{d}^{[i]} = [u_{d}^{[i]}(1) \,\, u_{d}^{[i]}(2) \,\, ...\,\, u_{d}^{[i]}(m)]^{T}$, and $u_{d}^{[i]}(j)$ is an arbitrary complex-valued weight. 

\section{BIA using Staggered Antenna Switching: DoF Characterization}
In this section, we derive an achievable sum DoF of the interference channel with BIA using staggered antenna switching at the receivers. In the next theorem, we assume no CSIT, each receiver is equipped with a reconfigurable antenna with an arbitrary number of antenna modes, and each transmitter has a conventional antenna.
\begin{thm}
{\it In the $K$-user SISO interference channel with BIA using staggered antenna switching, a sum DoF of $\frac{2K}{K+2}$ is achievable almost surely.}
\end{thm}
\begin{proof}
Assume that user $i$ sends $d_{i}$ symbols over $m$ channel uses. That is, ${\bf x}_{i} = \sum_{d=1}^{d_{i}} s_{d}^{[i]} \, {\bf u}_{d}^{[i]}$. Let ${\bf v}_{d}^{[i]} = s_{d}^{[i]} \, {\bf u}_{d}^{[i]}$ be the $d^{th}$ dimension transmitted by user $i$, and ${\bf V}^{[i]} = [{\bf v}_{1}^{[i]} \,\, {\bf v}_{2}^{[i]} \, ... \, {\bf v}_{d_{i}}^{[i]}]$ be an $m \times d_{i}$ matrix containing all dimensions transmitted by user $i$. Now assume the following alignment scheme which we will use to prove the achievability result: ${\bf v}_{d}^{[i]}$ aligns with $l-1$ dimensions from a set of $l-1$ distinct transmitters, e.g., $\{1,2,...,i-1,i+1,...,l\}$ at the remaining $(K-l)$ receivers\footnote{Note that a set of $l$ dimensions from $l$ transmitters can be aligned, at most, at $(K-l)$ receivers. Otherwise, some receivers will have their desired signals aligned with interference. Besides, we cannot align multiple dimensions from the same transmitter at any unintended receiver, because if we did so, those dimensions will align at their desired receiver and will not be decodable.}, e.g., $\{l+1, l+2,...,K-1,K\}$. In the rest of the proof, we obtain the optimal value for $l$. Note that the following condition is satisfied
\begin{equation}
\label{4}
\mbox{at receiver $j$:} \,\, {\bf H}_{ji} {\bf v}_{d}^{[i]} \in \mbox{span}({\bf H}_{jk} {\bf V}^{[k]}),
\end{equation}
$\forall j \in \{l+1, l+2,...,K-1,K\}$, and $\forall k \in \{1, 2,..., i-1,i+1,...,l-1,l\}$. Now assume that ${\bf v}_{d}^{[i]}$ also aligns with another $l$ dimensions from a set of $l$ distinct transmitters that include the transmitter $l+1$ instead of $l$, i.e., $\{1,2,...,i-1,i+1,...,l-1,l+1\}$ at the remaining $(K-l)$ receivers, i.e., $\{l, l+2, l+3,...,K-1,K\}$. In this case, the following condition is satisfied
\begin{equation}
\label{5}
\mbox{at receiver $j$:} \,\, {\bf H}_{ji} {\bf v}_{d}^{[i]} \in \mbox{span}({\bf H}_{jk} {\bf V}^{[k]}),
\end{equation}
$\forall j \in \{l, l+2, l+3,...,K-1,K\}$, and $\forall k \in \{1, 2,..., i-1,i+1,...,l-1,l+1\}$. Given that all channel matrices ${\bf H}_{ji}$, $\forall i$, have the same changing pattern dictated by the predefined antenna mode switching pattern at receiver $j$, we directly apply Lemma 2 in \cite{7} to (\ref{4}) and (\ref{5}) to obtain
\begin{align}
\label{6}
{\bf v}_{d}^{[i]} \in \mbox{span}({\bf V}^{[k]}),
\end{align}
$\forall k \in \{1, 2,..., i-1,i+1,...,l-1, l, l+1\}$. Now, consider receiver $l+1$. From (\ref{4}) we have ${\bf H}_{l+1,i} \, {\bf v}_{d}^{[i]} \in  \mbox{span}({\bf H}_{l+1,k} {\bf V}^{[k]})$, which can be written as
\begin{equation}
\label{8}
{\bf H}_{l+1,i} \, {\bf v}_{d}^{[i]} \in \mbox{span}({\bf H}_{l+1,l+1} {\bf H}^{-1}_{l+1,l+1} {\bf H}_{l+1,k} {\bf V}^{[k]}).
\end{equation}
Thus, we have ${\bf v}_{d}^{[i]} \in \mbox{span}({\bf H}^{-1}_{l+1,l+1} {\bf H}_{l+1,k} {\bf V}^{[k]})$, $\forall k \in \{1, 2,..., i-1,i+1,...,l-1, l\}$. From (\ref{6}), we know that ${\bf v}_{d}^{[i]} \in \mbox{span}({\bf V}^{[l+1]})$. Therefore, $\mbox{span}({\bf V}^{[l+1]})$ and $\mbox{span}({\bf H}^{-1}_{l+1,l+1} {\bf H}_{l+1,k} {\bf V}^{[k]})$ have at least one-dimensional common subspace, and $\mbox{dim}({\bf V}^{[l+1]} \bigcap {\bf H}^{-1}_{l+1,l+1} {\bf H}_{l+1,k} {\bf V}^{[k]}) > 0$, which implies that ${\bf H}_{l+1,l+1} \, {\bf V}^{[l+1]}$ and ${\bf H}_{l+1,k} \, {\bf V}^{[k]}$ have a non-zero intersection $\forall k \in \{1, 2,..., i-1,i+1,...,l-1, l\}$. Because ${\bf H}_{l+1,l+1} \, {\bf V}^{[l+1]}$ is the desired signal, receiver $l+1$ will have its signal polluted with interference if ${\bf v}_{d}^{[i]}$ is aligned with distinct dimensions from the set of transmitters $\{1,2,...,i-1,i+1,...,l-1,l+1\}$ and $\{1,2,...,i-1,i+1,...,l-1,l\}$ simultaneously. Therefore, we reach the following conclusion. {\it Given the adopted alignment scheme, if any dimension from one transmitter aligns with interference from $l-1$ distinct transmitters at the remaining $K-l$ unintended receivers, then it cannot be aligned with any other set of transmitters at another set of unintended receivers}.

Let $d_{i_{1},i_{2},...,i_{l}}$ denote the number of dimensions of the common subspace projected from the $l$ transmitters $\{i_{1},i_{2},...,i_{l}\} \subset \{1,2,...,K\}$ at the remaining $K-l$ receivers. Then we can generalize conditions (21) and (22) in \cite{7} as
\begin{align}
\label{101}
d_{j_{1},j_{2},...,j_{l}} = d_{u_{1},u_{2},...,u_{l}},
\end{align}
$\forall j_{1} \neq j_{2} \neq ... \neq j_{l}, u_{1} \neq u_{2} \neq ... \neq u_{l}$, and $j_{1}, j_{2},..., j_{l}, u_{1}, u_{2}, ..., u_{l} \in \{i_{1},i_{2},...,i_{l}\}$, and
\begin{align}
\label{10}
\sum_{\substack{k_{1}=1\\k_{1} \neq i}}^{K} \sum_{\substack{k_{2}=k_{1}+1\\k_{2} \neq i}}^{K} \ldots \sum_{\substack{k_{l-1}=k_{l-2}+1\\k_{l-1} \neq i}}^{K} d_{i,k_{1},...,k_{l-1}} \leq d_{i},
\end{align}
where $i \in \{1, 2,...,K \}$. Now, consider receiver $k$. Because the desired signal vectors must be linearly independent from each other and from interference, they occupy $d_{k}$ dimensions. In addition, the interfering signals from every set of $l$ interfering transmitters $\{i_{1},i_{2},...,i_{l}\} \subset \{1,2,...,k-1,k+1,...K\}$ overlap in $d_{i_{1},i_{2},...,i_{l}}$ dimensions. Because the total number of dimensions is equal to the number of channel uses $m$, we have
\begin{align}
\label{11}
\sum_{i=1}^{K} d_{i} - \sum_{\substack{k_{1}=1\\k_{1} \neq k}}^{K} \sum_{\substack{k_{2}=k_{1}+1\\k_{2} \neq k}}^{K} \ldots \sum_{\substack{k_{l}=k_{l-1}+1\\k_{l} \neq k}}^{K} d_{k_{1},k_{2},...,k_{l}} \leq m.
\end{align}
By summing the inequality in (\ref{11}) over $k$, we have
\begin{align}
\label{12}
K\sum_{i=1}^{K} d_{i} - (K-l)\sum_{k_{1}=1}^{K} \sum_{k_{2}=k_{1}+1}^{K} \ldots \sum_{k_{l}=k_{l-1}+1}^{K} d_{k_{1},k_{2},...,k_{l}} \leq K m,
\end{align}
where the factor $(K-l)$ arises from the fact that each dimension $d_{k_{1},k_{2},...,k_{l}}$ aligns at $(K-l)$ receivers. Next, we sum (\ref{10}) over $i$ to get $\sum_{i=1}^{K}\sum_{\substack{k_{1}=1\\k_{1} \neq i}}^{K} \sum_{\substack{k_{2}=k_{1}+1\\k_{2} \neq i}}^{K} \ldots \sum_{\substack{k_{l-1}=k_{l-2}+1\\k_{l-1} \neq i}}^{K} d_{i,k_{1},...,k_{l-1}}$. It can be shown that this series will contain $d_{i,k_{1},...,k_{l-1}}$, $\{i,k_{1},...,k_{l-1}\} \in \{1,2,...,K\}$ with all possible permutations of $\{i,k_{1},...,k_{l-1}\}$. For instance, at $K=5$ and $l=3$, the summation reduces to $d_{123}+d_{124}+d_{125}+d_{134}+d_{135}+d_{145}+d_{213}+d_{214}+d_{215}+d_{234}+d_{235}+d_{245}+d_{312}+d_{314}+d_{315}+d_{324}+d_{325}$. From (\ref{101}), we know that $d_{123}=d_{132}=d_{321}=...=d_{213}$, then we have $\sum_{i=1}^{K}\sum_{\substack{k_{1}=1\\k_{1} \neq i}}^{K} \sum_{\substack{k_{2}=k_{1}+1\\k_{2} \neq i}}^{K} \ldots \sum_{\substack{k_{l-1}=k_{l-2}+1\\k_{l-1} \neq i}}^{K} d_{i,k_{1},...,k_{l-1}} = l ! \sum_{k_{1}=1}^{K} \sum_{k_{2}=k_{1}+1}^{K} \ldots \sum_{k_{l}=k_{l-1}+1}^{K} d_{k_{1},k_{2},...,k_{l}}$. Substituting with inequality (\ref{10}), we have $l ! \sum_{k_{1}=1}^{K} \sum_{k_{2}=k_{1}+1}^{K} \ldots \sum_{k_{l}=k_{l-1}+1}^{K} d_{k_{1},k_{2},...,k_{l}} \leq \sum_{i=1}^{K} d_{i}$. Thus, from (\ref{12}), we get
\begin{align}
\label{13}
K\sum_{i=1}^{K} d_{i} - \frac{(K-l)}{l !} \sum_{i=1}^{K} d_{i} \leq K m.
\end{align}
Note that $l$ is a design parameter, so we are interested in selecting $l$ that maximizes the sum DoF $\frac{1}{m}\sum_{i=1}^{K} d_{i}$. We can rearrange (\ref{13}) as
\begin{align}
\label{14}
\frac{1}{m}\sum_{i=1}^{K} d_{i} \leq \sup_{l \geq 2} \frac{K l !}{K l ! -K +l} = \frac{2K}{K+2},
\end{align}
which is achieved at $l$ = 2.
\IEEEQEDhere
\end{proof}
Thus, for a large number of users, the achievable sum DoF of BIA in the $K$-user interference channel approaches $2$. This means that the best we can do using the proposed scheme when the CSIT is not available is to offer double the DoF achievable by orthogonal multiple access schemes. Note that unlike Theorem 7 in \cite{11}, which assumes a very specific channel coherence structure, our result holds for any (general) coherence structure. In the following section, we propose an algorithm to systematically generate the antenna switching patterns and the beamforming vectors such that the $\frac{2K}{K+2}$ sum DoF is achieved.

\section{BIA Algorithm}

\subsection{Antenna switching patterns and beamforming vectors generation}

For the proposed algorithm, we assume that each reconfigurable antenna has only two modes, i.e., ${\bf p}_{k}(j) \in \{1,2\}$ is the selected antenna mode at the $j^{th}$ channel use. The channel between the $i^{th}$ transmitter and the $k^{th}$ receiver belongs to the set $\{h_{ki}(1), h_{ki}(2)\}$, depending on the switching pattern ${\bf p}_{k}$. We also assume that the elements of the beamforming vectors are binary, thus $u_{d}^{[i]}(j) \in \{0,1\}$, i.e., we either activate or deactivate a certain symbol at a certain channel use based on its beamforming vector. Later in this section, we show that imposing these assumptions, in addition to the algorithm presented in this subsection, the $\frac{2K}{K+2}$ sum DoF can be achieved. Such assumptions have important hardware implications. For instance, the proposed algorithm operates with low cost reconfigurable antennas that have only two modes. Besides, beamforming is very simple and applied by activating or deactivating certain symbols at the transmitter. The idea behind the algorithm and the insights based on which it was constructed can be found in the proof of Theorem 2.

In order to achieve the DoF presented in Theorem 1, we need to satisfy the bounds in (\ref{10}) and (\ref{13}). It is obvious that the bound in (\ref{10}) is achieved if every dimension transmitted by user $i$ aligns with some other dimension from another user. Thus, each user needs to send $K-1$ symbols per $m$ channel uses. Moreover, the bound in (\ref{13}) is satisfied if we use the minimal number of channel uses that keeps the desired signals and the interference linearly independent. This is simply satisfied by setting $m = \frac{1}{2}(K+2)(K-1)$, which is the division of the total number of symbols $K(K-1)$ by the maximum sum DoF. The following algorithm can be used to generate the transmit beamforming vectors and the antenna switching patterns for all users.

\begin{algorithmic}[1]
\Procedure{BIA}{$K$}
\State Generate a matrix ${\bf \tilde{P}} = [{\bf \tilde{p}}_{1}, \, {\bf \tilde{p}}_{2},...,{\bf \tilde{p}}_{K}] \in \mathbb{B}^{\frac{1}{2}(K+2)(K-1) \times K}$, where $\mathbb{B} = \{0,1\}$, such that all the column vectors that are the Hadamard product of all the combinations of $K-2$ column vectors in ${\bf \tilde{P}}$ are linearly independent.
\State Construct ${\bf P} = {\bf \tilde{P}} + {\bf 1}_{{\frac{1}{2}(K+2)(K-1) \times K}}$, where ${\bf 1}_{n \times m}$ is an $n \times m$ matrix with all elements = 1.
\State $x \gets 1$
\While{$x \leq \binom {K}{K-2}$}
\State Pick a new pair of distinct users $i$ and $j$, and a new pair of dimensions $k_{i}$ and $k_{j}$.
\State Set ${\bf u}^{[i]}_{k_{i}} = {\bf u}^{[j]}_{k_{j}} \gets {\bf \tilde{p}}_{1} \circ {\bf \tilde{p}}_{2} \circ ... \circ {\bf \tilde{p}}_{i-1} \circ {\bf \tilde{p}}_{i+1} \circ ... \circ {\bf \tilde{p}}_{j-1} \circ {\bf \tilde{p}}_{j+1} \circ ... \circ {\bf \tilde{p}}_{K}$, where $\circ$ is the Hadamard product (element-wise product).
\State $x \gets x+1$.
\EndWhile\label{euclidendwhile}
\State \textbf{return} ${\bf P}$ and ${\bf u}^{[1]}_{1},...,{\bf u}^{[1]}_{K-1},{\bf u}^{[2]}_{1},...,{\bf u}^{[2]}_{K-1},...,$ \\ ${\bf u}^{[K]}_{1},...,{\bf u}^{[K]}_{K-1}$.
\EndProcedure
\end{algorithmic}
The algorithm returns ${\bf P}$, whose column vectors represent the antenna switching patterns for all users, and ${\bf u}^{[i]}_{k}$, which is the beamforming vector for the $k^{th}$ dimension (symbol) sent by the $i^{th}$ user. In the next subsection, we investigate the DoF achieved by the proposed algorithm.

\subsection{DoF achievability by the proposed BIA algorithm}
Before investigating the achievable DoF by the proposed algorithm, we prove a useful lemma.
\begin{lem}
{\it For any $K$, there always exists a matrix ${\bf \tilde{P}} \in \mathbb{B}^{\frac{1}{2}(K+2)(K-1) \times K}$, such that the matrix ${\bf U}$ comprising column vectors that are the Hadamard product of all the combinations of $K-2$ column vectors in ${\bf \tilde{P}}$ has full column rank.}
\end{lem}
\begin{proof}
We can easily prove this by providing an example for a matrix that always satisfies this condition. Let ${\bf \tilde{P}} = [{\bf 1}_{K \times K}-{\bf I}_{K \times K}, {\bf A}_{(\frac{1}{2}K-1) \times (K+1) \times K}]$, where ${\bf A}_{(\frac{1}{2}K-1) \times (K+1) \times K}$ has distinct rows, with each row containing exactly $K-2$ ones. For instance, for $K = 4$, the matrix ${\bf \tilde{P}}$ is
\[{\bf \tilde{P}}^{T} = \begin{bmatrix}
     0 & 1 & 1 & 1 & 0 & 0 & 1 & 1 & 0 \\[0.3em]
		 1 & 0 & 1 & 1 & 0 & 1 & 0 & 1 & 1 \\[0.3em]
		 1 & 1 & 0 & 1 & 1 & 0 & 0 & 0 & 1 \\[0.3em]
		 1 & 1 & 1 & 0 & 1 & 1 & 1 & 0 & 0
     \end{bmatrix}.\]
Now focus on the matrix ${\bf U} = [{\bf u}_{1} {\bf u}_{2} ... {\bf u}_{\binom{K}{K-2}}]$, where each column vector ${\bf u}_{i}$ is the Hadamard product of a distinct subset of $K-2$ column vectors in ${\bf \tilde{P}}$.
It can be easily shown that all columns in ${\bf U}$ are independent, which directly follows from the fact that the last $(\frac{1}{2}K-1) \times (K+1)$ elements in each column in ${\bf U}$ have a single non-zero elements in unique positions. Thus, no set of scalars $\{\alpha_{1}, \alpha_{2},...,\alpha_{\binom{K}{K-2}}\}$ that are not all zero exists such that $\alpha_{1} {\bf u}_{1} + \alpha_{2} {\bf u}_{2} + ... + \alpha_{\binom{K}{K-2}} {\bf u}_{\binom{K}{K-2}} = {\bf 0}_{\frac{1}{2}(K+2)(K-1)}$. Since matrix ${\bf \tilde{P}}$ can be generated for any $K$, the lemma follows. \IEEEQEDhere
\end{proof}

In the next theorem, we prove that the proposed scheme achieves the $\frac{2K}{K+2}$ sum DoF.
\begin{thm}
{\it Applying BIA using the switching patterns and beamforming vectors generated by the algorithm in Section IV-A, the sum DoF of $\frac{2K}{K+2}$ is achieved almost surely for any $K$.}
\end{thm}
\begin{proof}
We focus on an arbitrary receiver $j$. From (\ref{1}), we know that the received signal at such receiver is given by ${\bf y}_{j} = \sum_{i=1}^{K} {\bf H}_{ji} {\bf x}_{i}$. Based on the algorithm in Section IV-A, we know that each transmitter sends $K-1$ symbols over $\frac{1}{2}(K+2)(K-1)$ channel uses. Given the antenna switching patterns matrix ${\bf P}^{\frac{1}{2}(K+2)(K-1) \times K}$ and the beamforming vectors for all transmitters, we need to prove the following:
\begin{enumerate}
\item At receiver $j$, the signal space does not collapse, i.e., all desired signal vectors are independent.
\item The interference at receiver $j$ always resides in a dimension of $\frac{1}{2}K(K-1)$.
\item The interference space and the desired signal space do not intersect, i.e., all desired signal vectors are independent to all interfering signal vectors.
\end{enumerate}
Condition (1) implies that the $K-1$ desired signal vectors must reside in $K-1$ independent dimensions. Since the total number of dimensions is equal to the number of channel uses $m = \frac{1}{2}(K+2)(K-1)$, then condition (2) follows, i.e., the interference should reside in the remaining $\frac{1}{2}(K+2)(K-1) - (K-1) = \frac{1}{2}K(K-1)$ dimensions. Condition (3) implies that no interference vector aligns with any signal vector. If all these conditions are satisfied, then the DoF of a user $j$ is the ratio between the number of its interference-free signal dimensions to the total channel uses, i.e., $DoF_{j} = \frac{K-1}{\frac{1}{2}(K+2)(K-1)} = \frac{2}{K+2}$. Thus, the sum DoF for all users is $\frac{2K}{K+2}$, which is the achievable DoF stated in Theorem 1. We start by proving condition (1). From lemma 1, we know that all the distinct beamforming vectors ${\bf u}^{[j]}_{k}$, $\forall \, 1 \leq k \leq K-1$ are linearly independent, which means that the signal space will not collapse as it resides in $K-1$ dimensions. Thus, condition (1) is satisfied. Now, we prove that condition (2) is also satisfied. Note that each of the  $K-1$ interferers on receiver $j$ sends $K-1$ symbols. From step 7 in the algorithm, we know that any beamforming vector $k_{i}$ for interferer $i$ is the same as some beamforming vector $k_{w}$ of interferer $w$, i.e., ${\bf u}^{[i]}_{k_{i}} = {\bf u}^{[w]}_{k_{w}} \gets {\bf \tilde{p}}_{1} \circ {\bf \tilde{p}}_{2} \circ ... \circ {\bf \tilde{p}}_{i-1} \circ {\bf \tilde{p}}_{i+1} \circ ... \circ {\bf \tilde{p}}_{w-1} \circ {\bf \tilde{p}}_{w+1} \circ ... \circ {\bf \tilde{p}}_{K}$. Thus, the beamforming weights in both ${\bf u}^{[i]}_{k_{i}}$ and ${\bf u}^{[w]}_{k_{w}}$ are non-zero only if ${\bf \tilde{p}}_{j} = 1$, i.e., when receiver $j$ activates mode 2 of the reconfigurable antenna. Thus, the interfering dimensions $k_{i}$ and $k_{w}$ will always be aligned along the vector ${\bf u}^{[i]}_{k_{i}} = {\bf u}^{[w]}_{k_{w}}$ at receiver $j$ since all their non-zero weights perceive the same antenna mode, and will be independet from all other vectors as shown in lemma 1. Note that for interferer $i$, there also exists some beamforming vector $k^{'}_{i}$, which is the same as some beamforming vector $k^{'}_{j}$ of the desired user $j$ (we will show that they will not align while proving that condition (3) is satisfied). Thus, each interferer sends $K-1$ signal vectors, and has $K-2$ signal vectors each aligned with a distinct signal vector from another interferer, leading to a total of $(K-1)(K-2)/2$ interference dimensions. Besides, each interferer has a signal vector that does not align with any other interferer. Thus, the total number of interference dimensions is $(K-1)(K-2)/2 + (K-1) = \frac{1}{2}K(K-1)$. Thus, condition (2) is satisfied. Finally, we focus on condition (3). Recall that at receiver $j$, each interferer $i \neq j$ has some beamforming vector $k^{'}_{i}$, which is the same as some beamforming vector $k^{'}_{j}$ of the desired user $j$, i.e. ${\bf u}^{[i]}_{k^{'}_{i}} = {\bf u}^{[j]}_{k^{'}_{j}} \gets {\bf \tilde{p}}_{1} \circ {\bf \tilde{p}}_{2} \circ ... \circ {\bf \tilde{p}}_{i-1} \circ {\bf \tilde{p}}_{i+1} \circ ... \circ {\bf \tilde{p}}_{j-1} \circ {\bf \tilde{p}}_{j+1} \circ ... \circ {\bf \tilde{p}}_{K}$. Since both ${\bf \tilde{p}}_{i}$ and ${\bf \tilde{p}}_{j}$ are excluded from the Hadamard product, there exists channel uses at which the beamforming vectors of $i$ and $j$ are activated and the selected antenna mode is 1, and other channel uses  at which the beamforming vectors of $i$ and $j$ are activated and the selected antenna mode is 2. Thus, although ${\bf u}^{[i]}_{k^{'}_{i}} = {\bf u}^{[j]}_{k^{'}_{j}}$, the received desired signal vector and the interference vector do not align since the antenna pattern changes over the non-zero entries of the beamforming vectors, i.e. they are linearly independent almost surely. Hence, all signal vectors are independent on all interference vectors almost surely.
Since conditions (1), (2), and (3) are satisfied, then the DoF of user $j$ is $2/(K+2)$. By symmetry, the sum DoF is given by $2K/(K+2)$. Thus, the sum DoF in Theorem 1 is achieved.
\IEEEQEDhere
\end{proof}

\subsection{Application to the 4-user SISO Interference Channel}
Consider a fully connected 4-user SISO Interference Channel. It follows from the discussion in the previous section that each user can send 3 symbols over 9 channel uses. One possible choice for the matrix ${\bf P}$ is given by
\begin{equation}
\label{15}
{\bf P}^{T} = \begin{bmatrix}
     1 & 2 & 1 & 2 & 1 & 2 & 2 & 2 & 1 \\[0.3em]
		 2 & 1 & 1 & 2 & 2 & 1 & 2 & 2 & 2 \\[0.3em]
		 2 & 2 & 2 & 1 & 1 & 1 & 1 & 2 & 2 \\[0.3em]
		 2 & 2 & 2 & 2 & 2 & 2 & 1 & 1 & 1
     \end{bmatrix}
\end{equation}		
where it can be seen that ${\bf P}$ satisfies the conditions in lemma 1. The algorithm starts by picking any distinct pair of users, e.g., users 3 and 4, and selects any of their symbols, say the first symbol of each. These symbols can be aligned at receivers 1 and 2. We let the beamforming vectors ${\bf u}^{[3]}_{1}$ and ${\bf u}^{[4]}_{1}$ for both symbols be equal and set their values to $ {\bf \tilde{p}}_{1} \circ {\bf \tilde{p}}_{2}$. From (\ref{15}), we have ${\bf \tilde{p}}_{1} = [0 \,\, 1 \,\, 0 \,\, 1 \,\, 0 \,\, 1 \,\, 1 \,\, 1 \,\, 0]^{T}$, and ${\bf \tilde{p}}_{2} = [1 \,\, 0 \,\, 0 \,\, 1 \,\, 1 \,\, 0 \,\, 1 \,\, 1 \,\, 1]^{T}$. Thus, we have ${\bf u}^{[3]}_{1} = {\bf u}^{[4]}_{1} = [0 \,\, 0 \,\, 0 \,\, 1 \,\, 0 \,\, 0 \,\, 1 \,\, 1 \,\, 0]^{T}$. Similarly, we can align every pair of 2 symbols from 2 distinct users at the $K-2$ receivers of the remaining users by choosing the beamforming vectors as follows:
\begin{align}
\label{16}
{\bf u}^{[2]}_{1} = {\bf u}^{[4]}_{2} &= {\bf \tilde{p}}_{1} \circ {\bf \tilde{p}}_{3} =  [0 \,\, 1 \,\,0 \,\, 0 \,\,	0 \,\, 0 \,\,	0 \,\,	1 \,\,0]^{T}, \\
{\bf u}^{[2]}_{2} = {\bf u}^{[3]}_{2} &= {\bf \tilde{p}}_{1} \circ {\bf \tilde{p}}_{4} =  [0 \,\,	1 \,\, 0 \,\,	1 \,\, 0 \,\,	1 \,\,	0 \,\, 0 \,\,	0]^{T}, \\
{\bf u}^{[1]}_{1} = {\bf u}^{[4]}_{3} &= {\bf \tilde{p}}_{2} \circ {\bf \tilde{p}}_{3} =  [1 \,\,	0 \,\, 0 \,\,	0 \,\, 0 \,\,	0 \,\, 0 \,\,	1 \,\, 1]^{T} ,\\
{\bf u}^{[1]}_{2} = {\bf u}^{[3]}_{3} &= {\bf \tilde{p}}_{2} \circ {\bf \tilde{p}}_{4} =  [1 \,\,	0  \,\,	0  \,\,	1 \,\, 1 \,\,	0 \,\, 0 \,\,	0 \,\, 0]^{T},\\
{\bf u}^{[1]}_{3} = {\bf u}^{[2]}_{3} &= {\bf \tilde{p}}_{3} \circ {\bf \tilde{p}}_{4} =  [1 \,\, 1 \,\, 1 \,\, 0 \,\, 0 \,\, 0 \,\, 0 \,\, 0 \,\, 0]^{T},\\
{\bf u}^{[3]}_{1} = {\bf u}^{[4]}_{1} &= {\bf \tilde{p}}_{1} \circ {\bf \tilde{p}}_{2} =  [0 \,\, 0 \,\, 0 \,\, 1 \,\, 0 \,\, 0 \,\, 1 \,\, 1 \,\, 0]^{T}.
\end{align}

\begin{figure*}[!t]
\normalsize
\setcounter{mytempeqncnt}{\value{equation}}
\setcounter{equation}{20}
\begin{align}
\label{20}
{\bf H}_{1k} = \mbox{diag}([h_{1k}(1) \,\, h_{1k}(2) \,\, h_{1k}(1) \,\, h_{1k}(2) \,\, h_{1k}(1) \,\, h_{1k}(2) \,\, h_{1k}(2) \,\, h_{1k}(2) \,\, h_{1k}(1)]).
\end{align}
\setcounter{equation}{\value{mytempeqncnt}+1}
\hrulefill
\vspace*{4pt}
\end{figure*}

\begin{figure*}[!t]
\normalsize
\setcounter{mytempeqncnt}{\value{equation}}
\setcounter{equation}{21}
\[{\bf y}_{1} =\]
\begin{equation}
\label{21}
\underbrace{\begin{bmatrix}
h_{11}(1) & h_{11}(1) & h_{11}(1) \\
0 & 0  & h_{11}(2) \\
0	& 0	 & h_{11}(1) \\
0	& h_{11}(2) &	0 \\
0	& h_{11}(1) &	0 \\
0	& 0 &	0 \\
0 &	0	& 0 \\
h_{11}(2) &	0 &	0 \\
h_{11}(1) &	0 &	0
\end{bmatrix}}_{\mbox{{\bf Rank = 3}}}
\begin{bmatrix}
s_{1}^{[1]}\\
s_{2}^{[1]} \\
s_{3}^{[1]}
\end{bmatrix} +
\underbrace{\begin{bmatrix}
 0          & 0 				& h_{12}(1) & 0 				& 0 				& h_{13}(1) & 0 				& 0 				& h_{14}(1) \\
 h_{12}(2)  & h_{12}(2) & h_{12}(2) & 0 				& h_{13}(2) & 0 				& 0 				& h_{14}(2) & 0\\
 0 		   		& 0 				& h_{12}(1) & 0 				& 0 				& 0 				& 0 				& 0 				& 0\\
 0 					& h_{12}(2) & 0 				& h_{13}(2) &	h_{13}(2) &	h_{13}(2) &	h_{14}(2) &	0 				&	0\\
 0 					&	0  				&	0 				&	0 				&	0 				&	h_{13}(1) &	0 				&	0 				&	0\\
 0					&	h_{12}(2) &	0 				&	0 				&	h_{13}(2) &	0 				&	0 				&	0 				&	0\\
 0 					&	0  				&	0 				&	h_{13}(2) &	0 				&	0 				&	h_{14}(2) &	0				  &	0\\
 h_{12}(2) 	&	0 				& 0 				&	h_{13}(2) &	0					& 0 				&	h_{14}(2) &	h_{14}(2) &	h_{14}(2)\\
 0 					&	0  				&	0 				&	0 				&	0 				&	0 				&	0 				&	0				  &	h_{14}(1)
\end{bmatrix}}_{\mbox{{\bf Rank = 6}}}
\begin{bmatrix}
s_{1}^{[2]}\\
s_{2}^{[2]}\\
s_{3}^{[2]}\\
s_{1}^{[3]}\\
s_{2}^{[3]}\\
s_{3}^{[3]}\\
s_{1}^{[4]}\\
s_{2}^{[4]}\\
s_{3}^{[4]}
\end{bmatrix}
\end{equation}
\setcounter{equation}{\value{mytempeqncnt}+1}
\hrulefill
\vspace*{4pt}
\end{figure*}

\begin{figure*}[!t]
\normalsize
\setcounter{mytempeqncnt}{\value{equation}}
\setcounter{equation}{22}
\begin{equation}
\label{22}
{\bf R} \equiv
\begin{bmatrix}
 h_{11}(1) & h_{11}(1) & h_{11}(1) & 0          & 0 				& h_{12}(1) & 0 				 & h_{13}(1)   & h_{14}(1) \\
 0 				 & 0  			 & h_{11}(2) & h_{12}(2)  & h_{12}(2) & h_{12}(2) & 0 				 & 0 				   & 0\\
 0				 & 0	 			 & h_{11}(1) & 0 		   		& 0 				& h_{12}(1) & 0 				 & 0 				   & 0\\
 0				 & h_{11}(2) & 0 			   & 0 					& h_{12}(2) & 0 				& h_{13}(2)  & h_{13}(2)   & 0\\
 0				 & h_{11}(1) & 0 			   & 0 					&	0 				&	0 				&	0 				 & h_{13}(1)   & 0\\
 0				 & 0 				 & 0 			   & 0 					&	h_{12}(2) &	0 				&	0 				 & 0 				   & 0\\
 0 				 & 0				 & 0 		  	 & 0 					&	0 				&	0 				&	h_{13}(2)  & 0 				   & 0\\
 h_{11}(2) & 0 				 & 0 				 & h_{12}(2) 	&	0					& 0 				&	h_{13}(2)  & 0 				   & h_{14}(2)\\
 h_{11}(1) & 0 				 & 0  			 & 0 					&	0 				&	0 				&	0 				 & 0 				   & h_{14}(1)
\end{bmatrix}
\end{equation}
\setcounter{equation}{\value{mytempeqncnt}+1}
\hrulefill
\vspace*{4pt}
\end{figure*}

Now, we focus on receiver 1. The antenna switching pattern of receiver 1 is ${\bf p}_{1} = [1 \,\, 2 \,\, 1 \,\, 2 \,\, 1 \,\, 2 \,\, 2 \,\, 2 \,\, 1]^{T}$. Thus, the equivalent channel between transmitter $k$ and receiver 1 is given by (\ref{20}), and the received signal at receiver 1 is given by (\ref{21}), where we drop the additive noise term for convenience. From (\ref{21}), we see that the desired symbols $s_{1}^{[1]}$, $s_{2}^{[1]}$, and $s_{3}^{[1]}$ occupy a 3-dimensional space. For the interfering symbols, we observe that $s_{1}^{[2]}$ and $s_{2}^{[4]}$ align along the vector $[0 \,\, 1 \,\, 0 \,\, 0 \,\, 0 \,\, 0 \,\, 0 \,\, 1 \,\, 0]^{T}$, $s_{2}^{[2]}$ and $s_{2}^{[3]}$ align along the vector $[0 \,\, 1 \,\, 0 \,\, 1 \,\, 0 \,\, 1 \,\, 0 \,\, 0 \,\, 0]^{T}$, while $s_{1}^{[3]}$ and $s_{1}^{[4]}$ align along the vector $[0 \,\, 0 \,\, 0 \,\, 1 \,\, 0 \,\, 0 \,\, 1 \,\, 1 \,\, 0]^{T}$. This is due to the fact that every pair of these symbols have a non-zero entry in the their beamforming vectors only when the receiver selects mode 2, which is guaranteed by step 7 in the BIA algorithm. Thus, the interfering signal occupies a 6-dimensional space. To make sure that all desired symbols are decodable, we need to prove that all vectors carrying the desired symbols are linearly independent and the 3-dimensional subspace carrying the desired symbols does not intersect with the 6-dimensional interference subspace. This can be easily proved by showing that the 9 $\times$ 9 matrix ${\bf R}$ defined in (\ref{22}), which contains all the received desired and interference vectors, is full rank, which follows from the fact that all rows and columns in ${\bf R}$ are linearly independent almost surely. Thus, receiver 1 achieves $\frac{1}{3}$ DoF almost surely. The same analysis can be applied to receivers 2, 3, and 4. At every receiver, the interference will occupy a 6-dimensional space while the desired signal occupies a 3-dimensional space. Hence, the achieved sum DoF is $\frac{4}{3}$.
Note that for $K$ = 3, each two symbols can be aligned at one unintended receiver, and the proposed algorithm reduces to the scheme in \cite{7} achieving a sum DoF of $\frac{6}{5}$. For $K$ = 5, each two symbols can be aligned at 3 unintended receivers, and the achieved sum DoF is $\frac{10}{7}$.

\section{Conclusion}
In this paper, we have shown that linear Blind Interference Alignment (BIA) using staggered antenna mode switching can achieve a sum DoF of $\frac{2K}{K+2}$ in the $K$-user SISO interference channel. A key insight is that each signal dimension from one user can be aligned with a single set of distinct users at the receivers of the remaining users. This result suggests that we can double the unity DoF of orthogonal multiple access schemes without channel state information at the transmitters. Moreover, we proposed an algorithm to generate the transmit beamforming vectors and antenna switching patterns utilized in BIA. We showed that the proposed algorithm can achieve the $\frac{2K}{K+2}$ sum DoF for any $K$. By applying this algorithm to the 4-user Interference Channel, it was shown that a sum DoF of $\frac{4}{3}$ is achievable. 


\end{document}